%% file: main_arxiv.tex
\documentclass[letterpaper]{article} 
\usepackage[preprint]{aaai2027}  
\usepackage[hyphens]{url}  
\usepackage{graphicx} 
\usepackage{natbib}  
\usepackage{caption} 
\usepackage{booktabs}
\usepackage{amsmath,amssymb,amsthm}
\input{figures/figstyle}

\newtheorem{definition}{Definition}
\newtheorem{lemma}{Lemma}
\newtheorem{proposition}{Proposition}
\newtheorem{corollary}{Corollary}

\newcommand{\Ecomp}{P^{\mathrm{comp}}}

\title{Collusion with Competitive Marginals:\\Price-Level Audits Are Blind by Construction}

\author{
    Xin Xu\textsuperscript{\rm 1},
    Chengrui Wu\textsuperscript{\rm 2},
    Jiayu Lu\textsuperscript{\rm 1},
    Kaizhen Tan\textsuperscript{\rm 1},
    Siru Tao\textsuperscript{\rm 1},
    Hanzhe Hong\textsuperscript{\rm 1}
}
\affiliations{
    \textsuperscript{\rm 1}Carnegie Mellon University\\
    \textsuperscript{\rm 2}Zhejiang University\\
    xuxin@cmu.edu, chengrui.24@intl.zju.edu.cn, stellalu@tepper.cmu.edu,\\
    kaizhent@cmu.edu, sirutao@andrew.cmu.edu, hanzheho@andrew.cmu.edu
}

\begin{document}

\maketitle

\begin{abstract}
Empirical work on algorithmic collusion asks one question of the data: are prices
supracompetitive? We show that this question can be answered ``no'' by a conspiracy
that is nonetheless profitable. Consider bidding agents that couple only through the
joint distribution of their unexplained bid components, leaving every agent's own bid
law exactly at the competitive law. Any test whose input is a single agent's price or
bid history then has power exactly equal to its false-positive rate, for every
coupling strength up to comonotonicity. The published detection methodology is
therefore blind to this conduct by construction rather than underpowered, and no
sample size repairs it.

Three empirical results follow. First, the mechanism appears in real language-model
agents: twenty models from nineteen independent developers, three deployment prompts
each, show residual correlation of $+0.053$ between two deployments of one model
against $+0.0001$ across models, with a $95\%$ interval clustered by developer of
$[0.030,0.078]$, under an auditor that sees every order feature and is fitted out of
sample. Second, the coupling falls monotonically as sampling temperature rises
($p = 0.002$), turning a deployment parameter into a candidate mitigation. Third, on 24 days of
Ethereum block-building auction data covering 77{,}684 bids from 39 bidders, the
honest population of bidder pairs is itself so dependent that a screen held at a
5\% false-positive rate must sit above a floor of $+0.50$ to $+0.81$, which is 20 to
32 times the family-wise sampling threshold and does not fall as the audit window
grows. Since
lawful multi-identity operation and conspiracy are behaviourally indistinguishable
here, we argue that the tractable regulatory target is not detection but counting.
Measured concentration depends heavily on the counting rule: resolving 40 bidding
identities into 23 operators raises the Herfindahl index by $+247.5\%$, and adding
behavioural clusters recovered from public bid streams reaches $+324.5\%$.
\end{abstract}

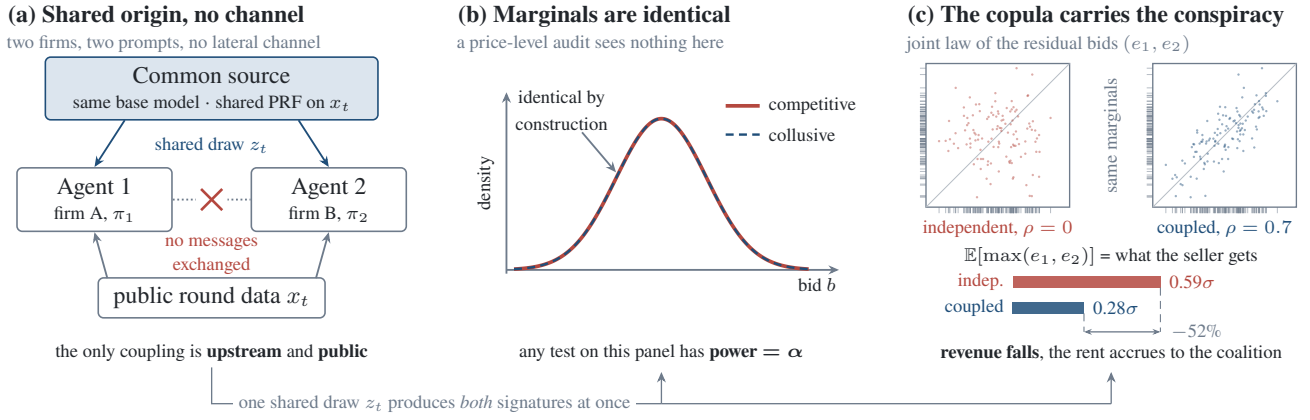
\begin{figure*}[t]
\centering
\input{figures/fig1_mechanism}
\caption{The mechanism. \textbf{(a)} Two agents deployed by different firms share an
upstream origin, either one base model under two deployment prompts or a
pseudorandom function keyed on public round data, and never exchange a message.
\textbf{(b)} Each agent's own bid distribution is identical under competitive and
collusive conduct, so every test that reads one agent's marginal has power equal to
its size. \textbf{(c)} The conspiracy lives entirely in the copula. Coupling the
unexplained components lowers the expected maximum bid, which is what the seller
receives, at unchanged marginal spread.}
\label{fig:mechanism}
\end{figure*}

\section{Introduction}

Two firms deploy bidding agents built on the same foundation model. The agents never
communicate. Each firm tunes its own deployment prompt, and each agent's bids, viewed
alone, are indistinguishable from those of a competitive bidder. Yet the seller
systematically receives less than it would from two independent bidders. Nothing in
the current empirical toolkit for algorithmic collusion would register this.

That toolkit measures price levels. Since \citet{calvano2020ai} showed that
independent Q-learning agents converge to supracompetitive prices, and a subsequent
literature extended the finding to language-model agents
\citep{fish2024llm,llmcollusion2024b}, the empirical question has been whether
observed prices exceed a competitive benchmark. Reported effects are large, on the
order of a 22\% price lift falling to 7 to 10\% under agent heterogeneity
\citep{heterogeneity2026}. Regulatory proposals inherit the framing and audit
algorithms one at a time \citep{hartline2024auditing,hartline2025auditing}.

This paper identifies a profitable conduct that the framing cannot see. Fix each
agent's marginal bid distribution at its competitive value and change only the
dependence between agents. Seller revenue in a first-price setting is the maximum
bid, and the expected maximum of a fixed set of marginals is minimised by the
comonotone coupling, so coordinating the joint law alone transfers rent. Because the
marginals never move, there is no price lift anywhere in the data. The construction
needs no communication channel: agents that evaluate a common function of public
round data are already coupled, and shared or merely similar model families supply
that common function for free \citep{kleinberg2021monoculture,oraclefingerprint}.

Two features make this awkward for the tools built to catch collusion. There is no
agreement, so the doctrine that treats a meeting of the minds as the object of
enforcement has nothing to point at, and the two firms need not know their agents are
coupled at all. And the party that gains need not be either firm: whoever supplies the
shared model fixes the coupling, so the rent can accrue to a bidder-adjacent third
party while the nominal competitors bid in good faith. The harm survives the removal
of every element these tools were designed around, an agreement to find and a price
signal to measure.

Our central claim is a statement about detection rather than about collusion. Let a
\emph{price-level test} be any test whose input is one agent's marginal bid or price
distribution compared against a competitive benchmark. Every empirical detection
method we are aware of in the algorithmic-collusion literature is of this form. Against a conduct
that fixes marginals, every such test has power exactly $\alpha$, its own size. This
is not low power that more data would fix. It is an identity, and it holds up to
perfect coupling.

\paragraph{Contributions and evidence.}
The formal content is short and its proof is one line. What makes it more than a
curiosity is that the conduct it describes is profitable at magnitudes comparable to
the effects the price-level literature reports, that it appears in real
language-model agents, and that a deployed market already sits at the dependence
level where it bites.

\begin{itemize}
\item A construction with fixed marginals and zero communication, and the exact
      powerlessness of marginal audits against it, verified in simulation across
      coupling strengths and sample sizes up to $2\times10^{4}$
      (Figure~\ref{fig:power}).
\item Replication with real language-model agents. Across nineteen independent
      developers, two deployments of one model couple ($+0.053$, clustered $95\%$
      interval $[0.030,0.078]$) while different models do not ($+0.0001$), and the
      contrast survives an auditor that sees every observable and is fitted out of
      sample.
\item Sampling temperature as a candidate dial: the coupling weakens as temperature
      rises, consistently across model families though not at conventional
      significance with five temperature points. Temperature is a deployment parameter
      a regulator could mandate, which makes it worth testing at scale.
\item A market calibration on 24 days of Ethereum MEV-Boost relay data. The binding
      constraint on detection is the spread of the honest dependence population, not
      sampling error, so the smallest detectable coupling does not vanish as the
      audit window grows.
\item A positive result. Conspiracy and lawful multi-identity operation leave the
      same behavioural trace here, so we measure how many independent decision-makers
      a market contains rather than trying to separate the two. Entity resolution
      raises measured concentration by $+247.5\%$, and behavioural clustering from
      public bid streams reaches $+324.5\%$.
\end{itemize}

\section{Related Work}

\paragraph{Algorithmic collusion.} That learning agents reach supracompetitive prices
without communicating is established for reinforcement learners \citep{calvano2020ai}
and for language-model agents \citep{fish2024llm,llmcollusion2024b}, the moderating
effect of heterogeneity has been quantified \citep{heterogeneity2026}, and the policy
debate is surveyed by \citet{algocollusionsurvey}. We claim none of this. Our claim
concerns what those studies measure: all of them read price levels, and we exhibit
conduct that leaves price levels untouched.

\paragraph{Screening.} Pairwise correlation of residual bids is standard cartel
screening \citep{harrington2008detecting,oecd2022screening}. We claim not the statistic
but a property of the population it is applied to. A universal impossibility claim would
also be false, since \citet{chassang2022robust} construct observational screens with
real power against noncompetitive bidding; our result is narrower and concerns tests
reading one agent's marginal.

\paragraph{Auditing.} \citet{hartline2024auditing} and \citet{hartline2025auditing}
give regulator-side audits defining non-collusion as low calibrated (swap) regret in
each agent's own transcript. These are not price-level tests, so
Proposition~\ref{prop:blind} does not cover them; we show below that they do not
separate the construction either, for a different and more conditional reason.

\paragraph{Correlated models.} \citet{kleinberg2021monoculture} show that shared
algorithms can lower collective decision quality with no correlated shock, and
\citet{oraclefingerprint} measure correlated forecasting errors across independently
developed models. Both establish that shared or similar model origins produce
correlated outputs, so our construction need not assume the dependence into existence.

\paragraph{Coupling and market.} That the comonotone coupling is extremal in the
supermodular order is classical \citep{muller2002comparison,puccetti2015extremal} and
Proposition~\ref{prop:comonotone} claims no novelty. Bidding incentives to compete or
collude in Ethereum block building are documented \citep{mevbuildercollusion}, and that
market's structure has been measured \citep{mev2305,mev2605}; we use it as a
measurement setting, not a new object of study.

\section{Setup}

Rounds $t=1,\dots,N$. In round $t$ a set $\mathcal{M}_t \subseteq \{1,\dots,M\}$ of
bidders competes for one indivisible object, the highest bid wins, and the seller
receives it. Write bidder $i$'s bid as
\[
  b_{it} = g_{it} + \sigma \varepsilon_{it},
  \qquad
  g_{it} = \mu_t + a_i + \beta_i^{\top} x_t,
\]
where $\mu_t$ is the round effect, $a_i$ a persistent bidder effect, $\beta_i^\top
x_t$ bidder-specific loadings on public covariates, and $\varepsilon_{it}$ the part
no public information explains, with marginal law $F_i$. An auditor reconstructs
$g_{it}$ from public data and cannot reconstruct $\varepsilon_{it}$. Write
\[
  \kappa = \frac{\operatorname{sd}(\sigma\varepsilon)}{\operatorname{sd}(g)}
\]
for the unexplainable share of cross-bidder dispersion.

\begin{definition}[Conduct]
A \emph{conduct profile} is a joint law $P$ for
$(\varepsilon_{1t},\dots,\varepsilon_{Mt})$ whose marginals are the prescribed
$F_1,\dots,F_M$. \emph{Competitive conduct} $\Ecomp$ is the product measure. A
\emph{conspiracy on coalition} $S$ is a conduct under which
$(\varepsilon_{it})_{i\in S}$ is dependent, with non-members independent.
\end{definition}

The marginals are pinned under both hypotheses. A conspiracy therefore changes the
copula \citep{copulanelsen} and nothing else. Everything below follows from that one
modelling choice.

\begin{lemma}[Zero-communication realisation]
\label{lem:construction}
Let $k$ be a key shared by the coalition $S$, let $\mathrm{PRF}$ be a pseudorandom
function, and let $\mathrm{id}_t$ be public round data. Put $v_t =
\mathrm{PRF}_k(\mathrm{id}_t)/2^{\lambda}$ and $z_t = \Phi^{-1}(v_t)$, and for $i \in
S$ set
\[
  \varepsilon_{it}
  = F_i^{-1}\!\Big(\Phi\big(\sqrt{\rho}\, z_t + \sqrt{1-\rho}\, e_{it}\big)\Big),
  \qquad e_{it} \sim N(0,1) \ \text{i.i.d.}
\]
Then $\varepsilon_{it} \sim F_i$ exactly for every $\rho \in [0,1]$, and no message
is exchanged, since each member computes $z_t$ from $k$ and public data alone.
Consequently $\rho=0$ reproduces $\Ecomp$ and $\rho=1$ gives the comonotone coupling
on $S$.
\end{lemma}

\begin{proof}
If $z,e \sim N(0,1)$ are independent then $\sqrt{\rho}z + \sqrt{1-\rho}e \sim
N(0,1)$, so its $\Phi$-image is uniform on $(0,1)$ and its $F_i^{-1}$-image is $F_i$.
The remaining claims are immediate.
\end{proof}

\paragraph{What the pseudorandom function does not buy.} It removes the need for a
channel and makes the shared draw unpredictable without $k$, so an outsider cannot
say in advance which rounds are favoured. It does not conceal the dependence. If
$v_t$ were genuinely uniform the coalition would still share $z_t$, so the joint law
remains statistically distinguishable from independence. We make no
computational-indistinguishability claim about the joint law, and the rest of the
paper concerns statistical thresholds precisely because none is available.

For language-model agents the shared draw needs no key. Two deployments of one base
model reading the same round data already compute correlated functions of it, which
is the coupling we measure in a later section.

\section{Marginal Audits Have Power Exactly $\alpha$}

\begin{proposition}
\label{prop:marginal}
Let $A$ be any test of size $\alpha$ under $\Ecomp$ whose input is a single bidder
$i$'s history $\{(b_{it}, x_t)\}_{t \le N}$. Under any conduct satisfying
Lemma~\ref{lem:construction},
\[
  \Pr\nolimits_{P}(A \text{ rejects}) = \alpha .
\]
\end{proposition}

\begin{proof}
The shared draws $z_t$ are independent across rounds, so under
Lemma~\ref{lem:construction} bidder $i$'s own sequence $(\varepsilon_{it})_t$ is
i.i.d.\ $F_i$, exactly as under $\Ecomp$: the whole single-bidder process law, and not
merely its per-round marginal, is unchanged. A measurable function of data with
identical laws has an identical law, and so does the rejection event.
\end{proof}

The proof is one line and the content is not in it. The content is that
Lemma~\ref{lem:construction} supplies a \emph{profitable} conduct with this property,
which the next proposition establishes. The statement is also stronger than low
power. Power equals size identically, for every $\rho$ up to $1$.

\paragraph{The independence of the shared draw is what buys this, and it is tight.}
Because the whole process law is at stake and not just the marginal, a single-bidder
test may look for serial structure as well as shape, so it is worth being precise about
what the construction needs. Simulating one bidder's own series confirms it is
indistinguishable from the competitive one at every $\rho$ up to comonotonicity, in
marginal shape and at every lag. Make the shared draw \emph{persistent} instead, an
AR(1) common signal rather than an independent one, and the bidder's own series
inherits the autocorrelation: at persistence $0.5$ its lag-one autocorrelation is
$0.40$, and at $0.9$ even a marginal goodness-of-fit test rejects. Serial dependence in
one agent's residuals is therefore evidence of a persistent common signal, which is a
different and detectable object. The impossibility is bought by the shared draw being
refreshed each round, and it is bought exactly there.

\begin{proposition}[Classical]
\label{prop:comonotone}
Among all couplings of $(\varepsilon_{it})_{i \in S}$ with marginals fixed at
$(F_i)_{i\in S}$, the comonotone coupling minimises $\mathbb{E}[\max_{i \in S}
b_{it}]$.
\end{proposition}

\begin{proof}
The coordinate maximum is submodular, and the comonotone coupling is the maximum of
the Fr\'echet class in the supermodular order, hence it minimises the expectation of
a submodular function \citep{muller2002comparison,puccetti2015extremal}.
\end{proof}

\paragraph{From the seller's loss to the coalition's gain.} Lowering
$\mathbb{E}[\max_{i \in S} b_{it}]$ is a statement about what the seller receives, and
it does not by itself imply the coalition captures the difference. With non-members
present, a coalition that bids lower also wins less often, and the seller's loss can
leak to an outsider instead. The step is therefore conditional, and the condition is
dispersion. If the bidders the coalition faces have no unexplained dispersion of their
own, the coupling only trades win probability against a deterministic threshold and can
leave the coalition worse off. In the regime this paper is about, where every bidder
carries unexplained dispersion, the two move together: across coalition sizes $2$ to
$5$, zero to six outsiders, and $\sigma$ from $0.05$ to $0.20$, coupling at $\rho=0.9$
raises coalition surplus in every configuration we checked while lowering seller
revenue, under a paired counterfactual in which $\rho=0$ returns exactly zero. The
quantity we measure on market data throughout is the seller-side loss, which is the
regulator's object in any case.

This corrects a natural but wrong intuition. Antithetic turn-taking, with members
cycling through who bids aggressively, is anti-profitable: for uniform marginals and
two bidders the expected maximum is $1/2$ comonotone, $2/3$ independent and $3/4$
antithetic. What costs a coalition money is the dispersion of its joint bids, so the
profitable move is to bid alike rather than to take turns.

\begin{figure}[t]
\centering
\input{figures/fig2_power}
\caption{Marginal audits have power equal to their size. Pairwise tests on the same
data reach power $1$, and the coupling strength they can resolve improves like
$N^{-1/2}$. Single-agent tests stay on the $\alpha=0.05$ line at every coupling
strength and every sample size. Monte Carlo, 200 trials per cell, 8 bidders, 2
conspirators.}
\label{fig:power}
\end{figure}
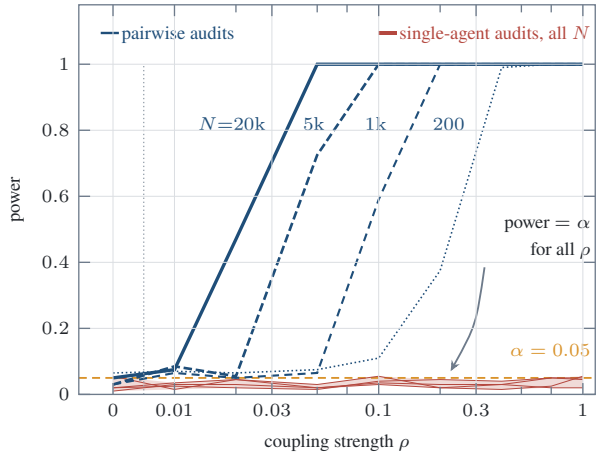

\paragraph{Simulation.} Six distributional statistics on one bidder's residuals,
spanning shape, serial dependence and spectrum, separate collusive from honest markets
at AUC $0.50$ to $0.55$ at $N$ up to $10^4$, and a Bonferroni battery of them has power
$0.020$ at $\rho = 1.0$ and $N = 2\times10^{4}$ against nominal $\alpha = 0.05$
(Figure~\ref{fig:power}). The pairwise test on the same data reaches power $1$.

\section{Why the Empirical Literature Cannot See This}

\begin{proposition}
\label{prop:blind}
Call a \emph{price-level test} any test whose input is the marginal distribution of
one agent's price or bid, compared against a competitive benchmark. Every such test
has power exactly $\alpha$ against a conduct satisfying
Lemma~\ref{lem:construction}.
\end{proposition}

Proposition~\ref{prop:blind} is a corollary of Proposition~\ref{prop:marginal}, and
its force comes from a fact about the literature rather than about mathematics. The
empirical detection methods we are aware of in algorithmic collusion ask whether
observed prices are supracompetitive. That is a question about a marginal. A conduct satisfying
Lemma~\ref{lem:construction} produces no supracompetitive price in any agent's
marginal, at any sample size, so the entire approach is blind to it by construction.

The magnitudes are what make this more than a curiosity. The price-level literature
reports lifts of order 22\%, falling to 7 to 10\% under heterogeneity
\citep{calvano2020ai,heterogeneity2026}. Measured below on real relay data, a
five-member coalition takes $4.6\%$ of seller revenue at the lower end of the cleaned
ambient floor and $10.6\%$ at the uncleaned one. These are the same order of magnitude,
and one is visible to the standard tests while the other is not.

\paragraph{The dependence does not have to be assumed.} It is documented.
\citet{kleinberg2021monoculture} establish that shared algorithms degrade collective
outcomes absent any correlated shock, and \citet{oraclefingerprint} measure
correlated errors at $r = 0.77$ to $0.78$ across three independently developed
models on forecasting tasks. What remains to establish empirically is not existence
but rate: how extractable rent varies with a parameter a regulator could set.

\paragraph{Regret-based audits do not separate it either.} The blindness is not
specific to price-level tests, though the reason is different and worth stating
carefully. The audits of \citet{hartline2024auditing} and \citet{hartline2025auditing}
read one agent's transcript and test its calibrated (swap) regret, calling an agent
non-collusive when no systematic remap of its own realised bids would have paid better.
A Lemma~\ref{lem:construction} conspiracy is \emph{not} exempt on best-response grounds:
it is not a correlated equilibrium, and it does leave a suppressed deviation.

\begin{proposition}[Not an equilibrium, but no excess regret]
\label{prop:regret}
Fix marginals with unexplained scale $\sigma>0$ and let $R(\rho)$ be the calibrated
swap regret of a Lemma~\ref{lem:construction} conduct at coupling $\rho$. Then
\emph{(i)} for every $\rho\in(0,1]$ the conduct is not a best response to its own shared
draw, so $R(\rho)>0$; in the comonotone limit the best response loads half the common
shift, $c^\ast=\tfrac12 v+\tfrac12\sigma z$, where the conduct bids $\tfrac12 v+\sigma
z$. \emph{(ii)} Best response to the shared draw contracts to the competitive
equilibrium, which earns zero rent and carries the competitive marginal, so no conduct
is at once a correlated equilibrium, profitable, and equal in marginal to $F_i$.
\emph{(iii)} The competitive conduct at the same marginals already carries regret
$R_{\mathrm{comp}}(\sigma)$ of order $\sigma^2$, and $R(\rho)\le R_{\mathrm{comp}}$ with
$R$ decreasing in $\rho$.
\end{proposition}

\begin{proof}[Proof of (i), $\rho=1$]
A rival playing Lemma~\ref{lem:construction} bids $\tfrac12 v'+\sigma z$ with $v'\sim
U[0,1]$, so $\Pr(\text{win at }c\mid z)=2(c-\sigma z)$ on the interior. Maximising
$(v-c)\cdot 2(c-\sigma z)$ gives $c^\ast=\tfrac12(v+\sigma z)$, which differs from the
prescribed bid by $\tfrac12\sigma z\neq 0$ on a positive-measure set, so the regret is
strictly positive. Part (ii) follows because the best-response loading on the shared
shift is a contraction to zero. Part (iii) is verified by a cross-fitted estimator
across coalition sizes, outsiders, noise scales and value laws (scripts released), and
left as a conjecture in general.
\end{proof}

Part (iii) is what saves the conduct from the audit, and it is a statement about the
honest market, not about the conspiracy. Because the unexplained component is bid
dispersion the auditor cannot rationalise, the competitive conduct itself carries
positive calibrated regret, and the copula does not \emph{raise} regret above that
level: under a bias-corrected (cross-fitted) estimator the excess is nonpositive and
falls with $\rho$. A regret audit calibrated to pass the ambient noisy-but-honest
market therefore also passes the conspiracy, and cannot separate the two. The
construction thus evades price-level and regret-based audits alike, but for distinct
reasons: the former because the marginal is untouched, the latter because the coupling
hides inside the market's own irreducible dispersion. This second claim is conditional:
it holds when that dispersion is off-best-response noise, which the measured $\kappa$
makes plausible without proving.

\section{Real Language-Model Agents}

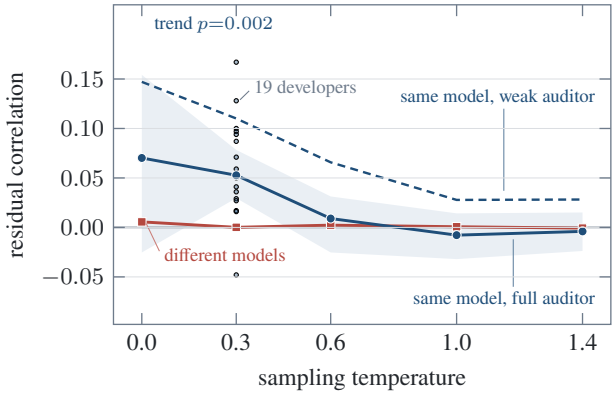
\begin{figure}[t]
\centering
\input{figures/fig3_temperature}
\caption{Coupling is base-specific and weakens with sampling temperature. Twenty models
from nineteen independent developers, three deployment prompts each, 200 rounds per
temperature. Two deployments of one model couple; two different models do not. The
solid line is the strict auditor, which sees every order feature and is fitted out of
sample, with a band clustered by developer; the dashed line is the same quantity under
an auditor seeing only two features. Dots give the nineteen per-developer means at
$T=0.3$, one of which is negative. The decline in temperature holds under every
auditor.}
\label{fig:temperature}
\end{figure}

We instantiate the construction with language-model agents bidding in a first-price
auction over synthetic orders. Twenty models from nineteen independent developers are
each deployed under three different prompts, giving sixty agents, run for 200 rounds at
each of five sampling temperatures. A pair is \emph{same-base} when its two agents are
two deployments of one model, and dependence is measured between pairs of residual bid
series after regressing out the order features an auditor is assumed to see. Because
that assumption does the work, we report every quantity under three auditors of
increasing strength: (A) linear in two of the six order features, (B) linear in all
six, and (C) all six with squares and interactions, fitted out of sample by five-fold
cross-fitting. Confidence intervals are clustered by developer, the level at which
agents are independent.

Same-base pairs couple and cross-base pairs do not, and the contrast survives the
strongest auditor. At temperature $0.3$ the sixty same-base pairs correlate at $+0.110$
under auditor (A), $+0.059$ under (B) and $+0.053$ under (C), with clustered $95\%$
interval $[0.030, 0.078]$ in the last case, while the $1710$ cross-base pairs sit at
$+0.023$, $+0.000$ and $+0.0001$. Attenuation across the ladder is expected, since a
richer auditor absorbs the part of the coupling that is common response to observable
order features; what does not attenuate is the asymmetry, because a fully specified
auditor drives cross-base dependence to zero and leaves same-base dependence
significantly positive. The coupling is base-specific in the way the shared-origin
account predicts, and it varies across developers, from $+0.17$ down to $-0.05$ under
auditor (C), so the mechanism is real on average rather than universal.

\paragraph{Temperature governs the coupling.} Sampling temperature is the agents' own
randomisation, and randomisation is what breaks a copula. Same-base correlation falls
monotonically with it, from $+0.147$ to $+0.028$ under auditor (A) and from $+0.070$ to
$-0.004$ under (C), as bid dispersion rises from $10.1$ to $18.1$. Permuting temperature
labels within developer over all pair-by-temperature observations, the mean within-family
slope is negative under every auditor, at $p = 5\times10^{-5}$ for (A), $p = 0.003$ for
(B) and $p = 0.002$ for (C). Temperature is a deployment parameter a regulator can
mandate, and raising it weakens the dependence the conduct requires. We do not read
rent off this ladder: the rent counterfactual permutes residuals, so its scale moves
with how the bid is decomposed, and only the coupling is comparable across auditors.

\paragraph{A contrast that sharpens the mechanism.} \citet{oraclefingerprint} find
cross-family correlation of $0.77$ to $0.78$ on forecasting, where ours in bidding is
indistinguishable from zero under a fully specified auditor. Forecasting has a ground truth that all competent models approach, so
they correlate through it, whereas bidding has none and idiosyncratic strategy
dominates. The coupling here is therefore base-specific rather than generic model
similarity, which is what the shared-origin account predicts.

\section{Market Calibration}

Classical asymptotics say the smallest detectable coupling shrinks like $N^{-1/2}$, so
a regulator need only audit for longer. A Fisher-$z$ test of $H_0 \colon \rho = 0$,
Bonferroni-corrected over all $P = \binom{M}{2}$ pairs, retains power $1/2$ down to
$\rho^{\ast}(N) \approx \sqrt{2\log(M^2/\alpha)/N}$, and since $\rho$ has finite
positive Fisher information at $0$ no level-$\alpha$ test does better than the
$N^{-1/2}$ rate. In simulation $\rho^\ast\sqrt{N}$ is flat at $3.10$ to $3.12$ for $N$
from $200$ to $2\times10^4$, matching the predicted constant. That reasoning fails in a
real market, and the next proposition says why.

\begin{figure}[t]
\centering
\input{figures/fig4_floor}
\caption{More data does not lower the detection floor. Observed 95th-percentile
pairwise residual correlation stays flat as the number of shared rounds grows, while
the permutation null collapses, and the same holds when rounds are subsampled within a
fixed set of pairs. The cleaned cross-operator floor is 20 to 32 times the family-wise
sampling threshold at $N = 2\times10^{4}$. Ethereum MEV-Boost relay
data, 24 days, 77{,}684 bids, 39 builders.}
\label{fig:floor}
\end{figure}
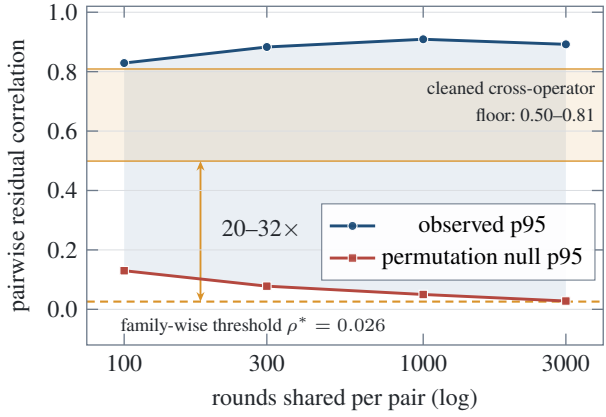

\begin{proposition}[Ambient floor]
\label{prop:floor}
Let $\mathcal{H}$ be the lawful-or-unattributed population of bidder pairs and
$F_{\mathcal{H}}$ the distribution of the pairwise dependence statistic over it. A
screen holding its per-pair false-positive rate at $\alpha$ on that population requires
threshold $\tau_\alpha = F_{\mathcal{H}}^{-1}(1-\alpha)$; controlling the family-wise
rate over $P$ pairs requires the corresponding quantile of the maximum, which is larger
still. Neither depends on $N$. Consequently the smallest detectable dependence is
bounded below by a constant and covert rent by $c \cdot \tau_\alpha > 0$, independent
of the audit window.
\end{proposition}

\begin{corollary}
The $N^{-1/2}$ law binds only when $\tau_\alpha \lesssim \rho^{\ast}(N)$, that is
when the honest population is nearly degenerate at independence. Otherwise the
constant binds.
\end{corollary}

\paragraph{Measurement setting.} We use Ethereum's MEV-Boost block-building auction
\citep{flashbots}, the only free, public and live market we know of in which both
sides of every quote carry a persistent identity. Builders are the bidders, proposers
the sellers, and the relay-published bid value is the bid. The panel spans 24 days,
77{,}684 bids, 5{,}408 slots and 39 builders. On it $\kappa = 0.895$, the ratio of
residual to explained standard deviation, so $44.5\%$ of cross-bidder dispersion is
left unexplained by public covariates.

The population is far from degenerate. The 95th percentile of observed pairwise
residual correlation is $+0.829$ at a 100-round overlap requirement and $+0.909$ at
1000, while the permutation null falls from $+0.130$ to $+0.050$, and to $+0.028$ at
3000. Raising that requirement also changes which pairs are scored, so on its own it
cannot separate sampling from composition. Holding the population fixed settles it:
among the $55$ pairs sharing at least $3000$ slots, subsampling rounds \emph{within}
those same pairs over a thirtyfold range moves the null from $+0.153$ to $+0.028$ while
the observed floor moves only from $+0.910$ to $+0.893$. On an identical set of pairs,
sampling error collapses and the dependence does not.

\paragraph{Cleaning the population.} Much of the raw floor is one operator running
several identities, which is lawful and common. Removing such pairs by the union of
two independent labellings, and deliberately excluding correlation-based clustering to
avoid circularity, leaves a floor of $+0.81$. Adding a submission-timing filter takes
it to $+0.50$. Measured against the sampling threshold for this market under the same
error criterion, that floor is $20$ to $32$ times the family-wise threshold and $36$ to
$58$ times the per-pair one; the two criteria must not be mixed, and neither may be read
off a different market size. Calibrating the family-wise threshold empirically instead,
from the permutation distribution of the largest pairwise correlation, puts it at
$+0.088$ against an observed 95th percentile of $+0.909$. We report a range rather than
a point because the timing filter does most of the work and rests on the premise that co-submission implies shared
infrastructure, which we assume rather than establish. If some co-submitting pairs
are genuine competitors, the lower end is biased downward.

\paragraph{Three repairs fail.} Entity resolution moves the floor only from $+0.916$ to
$+0.905$, though multi-identity operation is real: one operator runs six identities and
same-operator pairs sit at median dependence $+0.906$. Factor adjustment removing five
principal components carrying $85.8\%$ of variance moves the observed 95th percentile to
$+0.825$ while inflating the null from $+0.114$ to $+0.233$, losing power rather than
gaining validity. More data is ruled out by Proposition~\ref{prop:floor}, since
$\tau_\alpha$ contains no $N$.

\paragraph{The dichotomy, priced.} A coalition facing a submission-timing screen has
two options and both cost it. Remaining temporally separated leaves it screenable at
a floor of $+0.499$ rather than $+0.905$, which removes half its rent
(Table~\ref{tab:dichotomy}). Co-submitting to evade the timing screen makes it
present as one operator's identity cluster, forfeiting the claim to independent
competition, which is itself measurable as concentration.

\begin{table}[t]
\centering
\small
\begin{tabular}{lrrr}
\toprule
Coalition size & at $+0.905$ & at $+0.499$ & removed \\
\midrule
2 & 0.92\% & 0.46\% & 50.0\% \\
3 & 4.34\% & 1.81\% & 58.2\% \\
5 & 10.56\% & 4.57\% & 56.7\% \\
\bottomrule
\end{tabular}
\caption{Covert rent as a percentage of proposer revenue at the uncleaned and the
timing-screened floor, from a paired counterfactual on the archive panel at 103.2
ETH per day of proposer revenue.}
\label{tab:dichotomy}
\end{table}

The extractable share grows with the unexplainable dispersion $\kappa$ and with
coalition size, both larger in thinner or more concentrated markets than this one, so
the calibration is conservative rather than extreme.

\paragraph{The market already contains the signature.} The sharpest single
observation in the data needs no complete labelling. Five identity pairs bid at
correlation up to $+0.992$ while almost never taking top-of-book from each other, at
leapfrog rates of only $0.028$ to $0.060$. Each of those pairs belongs to a single
operator. This is exactly the signature that Lemma~\ref{lem:construction} produces,
occurring lawfully. A screen calibrated to catch the conspiracy would flag these
first.

\section{A Positive Result: Count Entities, Not Conduct}

Conspiracy and lawful multi-identity operation are not separated by the dependence
statistics available here, and we know of no behavioural statistic that separates
them, so a detection programme built on one is looking for a difference that may not
be in the data. What is decidable from public data is a different question: how many
independent decision-makers the market contains.

\paragraph{Detecting shared operation without labels.} Pairwise bid-residual
correlation recovers label-identified same-operator pairs at AUC $0.951$. Because the
24 positives come from only six operators, the honest interval is an
operator-clustered bootstrap, $[0.754, 0.993]$, and leave-one-operator-out gives
$0.926$ to $0.984$. The leave-one-pair-out envelope is far narrower but assumes an
independence the positives do not have. At threshold $0.80$ recall is $95.8\%$ at an
apparent false-positive rate of $5.4\%$. The detector survives the checks that matter:
permuting operator labels at fixed group sizes returns AUC $0.504$, so the pipeline does
not leak; restricting to pairs whose members both carry a label, which makes the
negative class genuinely negative, moves AUC \emph{up} to $0.957$; and overlap alone is
not the signal, detecting at AUC $0.434$. The value of a label-free detector is coverage
rather than accuracy, since the two labellings together reach only 24 of 48 builders
while correlation covers every bidder that submits.

\paragraph{Concentration depends on the counting rule.} Measured concentration on the
single-day archive panel moves substantially with how bidders are grouped. Over 40
bidding identities the Herfindahl index of bid leadership is $0.0655$. Grouping by
the union of the two operator labellings gives $0.2277$ over 23 entities, a rise of
$+247.5\%$. Adding behavioural clusters on top of the labels reaches $0.2782$ over 20
units, $+324.5\%$.

Behavioural clustering alone is the weakest of the three and is not a substitute for
labels. Its lift spans $+42\%$ to $+143\%$ across threshold and overlap settings, and
greedy complete linkage is order-dependent, giving $+63\%$ to $+189\%$ over tie-break
order alone, so it has no point estimate. Nor is it conservative in one direction: at
threshold $0.90$ it splits one operator's five identities while fusing two separately
labelled operators carrying $17.5\%$ of top bids. Against size-matched random groupings
the lift is still real ($p = 0.0008$), but the defensible claim is the combination,
since labels and behavioural clusters each find groupings the other misses. We report
bid leadership rather than realised market share, because the archive publishes relay
bids and not delivered payloads.

\paragraph{Dual use.} Lemma~\ref{lem:construction} is a recipe as well as an analysis.
The coupling it exploits arises by default from shared model supply rather than from
deliberate engineering, so the paper describes a condition markets are already in rather
than a capability they lack, and it confers no organisational advantage since it needs
neither communication nor agreement. The defensive content, that price-level evidence
cannot rule this conduct out and that entity counting is available from public data, is
actionable only if stated openly. We give no operational detail beyond what the
propositions require.

\paragraph{Reproducibility.} Market-side experiments run on CPU with numpy and scipy
alone, each an independent script writing a machine-readable result file, and every
market-side quantity plotted here is transcribed from one. Agents are queried through a
hosted inference API; relay and on-chain data come from public keyless endpoints. The
supplementary material holds every script, every result file, and the generated bids,
including the checks written solely to attack our own earlier results: the
coalition-surplus condition on Proposition~\ref{prop:comonotone}, the within-pair
subsampling separating sampling from composition, the family-wise threshold
calibration, and the process-law simulation behind
Proposition~\ref{prop:marginal}.

\section{Limitations}

\paragraph{The entity results are single-day, and the ceiling is structural.} The
labelling that grounds the detector and the concentration figures joins to builders
only through winning blocks, so a builder that never won cannot be labelled at all and
no amount of waiting fixes this. Pooling to a longer panel makes coverage worse, since
the wider feed samples few slots and rarely captures a winning block's hash, collapsing
the two labellings onto the same handful of winners. Every claim conditioning on
operator identity inherits this ceiling, against a labelling we cannot validate from a
third source.

\paragraph{The two halves are joined by argument, not measurement.} Block builders do
not run language models: the mechanism is demonstrated on language-model agents and the
ambient floor is measured on builders. We claim the combination as a calibration
showing that a deployed market sits where the mechanism would bite, not as one
experiment.

\paragraph{The lab rent is not comparable across auditors.} The rent counterfactual
permutes residuals, so its magnitude depends on how the bid is decomposed and moves by
an order of magnitude across the auditor ladder. We therefore report the coupling,
which is comparable, and treat the language-model experiment as evidence that the
mechanism is real rather than as an estimate of its size.

\paragraph{Thresholds are not derived.} The timing gap, the clustering thresholds and
the screen operating points come from monotone sweeps, reported as trade-off curves
rather than derived from a loss function.

\section{Conclusion}

The empirical study of algorithmic collusion asks whether prices are supracompetitive.
We have exhibited profitable conduct for which the answer is no by construction, shown
that every test of that form has power equal to its size against it, found that the
regret-based audits do not separate it either, replicated the mechanism in
language-model agents, and measured a deployed auction whose ambient dependence exceeds
the sampling threshold by a factor of 20 to 32. The two families fail for different
reasons: the marginal is untouched, and the coupling raises no regret above the honest
market's own. The consequence is not a better screen. It is that where lawful shared
infrastructure and conspiracy leave the same trace, the answerable question is how many
independent decision-makers a market holds, and that public data can answer.


\appendix
\setcounter{secnumdepth}{2}
\section*{Technical Appendix}
\noindent The material below records derivations and tables that the paper above states in compressed form. The paper is self-contained; nothing load-bearing is deferred here. Every table is regenerated by the named script in the code release.

\section{Auditor specifications for the language-model experiment}

Residual dependence is only defined relative to an auditor's model of the bid, so the
main paper reports the language-model result under three auditors of increasing
strength. All three regress each agent's bid series on order features and take
residuals; they differ in what the auditor is assumed to see and how flexibly it
fits.

\begin{description}
\item[(A)] Linear in two of the six order features (size, volatility).
\item[(B)] Linear in all six (size, volatility, urgency, complexity, competing
bidders, quoted spread).
\item[(C)] All six with squares and all pairwise interactions, 27 terms, ridge
penalty $10^{-3}n$, fitted by five-fold \emph{cross-fitting} so every residual is an
out-of-sample prediction error.
\end{description}

Auditor (C) is the strict one: it sees every observable an agent sees and cannot
manufacture dependence by overfitting, because residuals are held out. Table~\ref{tab:ladder}
gives the full ladder. Confidence intervals are $t$ intervals clustered by developer,
which is the level at which agents are independent; a same-base pair is two
deployments of one model, a cross-base pair spans two models.

\begin{table}[t]
\centering
\small
\footnotesize
\begin{tabular}{llrrr}
\toprule
Aud. & $T$ & same-base & clustered 95\% CI & cross \\
\midrule
(A) & 0.0 & $+0.147$ & $[+0.067, +0.228]$ & $+0.033$ \\
(A) & 0.3 & $+0.110$ & $[+0.069, +0.158]$ & $+0.023$ \\
(A) & 0.6 & $+0.066$ & $[+0.023, +0.102]$ & $+0.017$ \\
(A) & 1.0 & $+0.028$ & $[-0.004, +0.057]$ & $+0.016$ \\
(A) & 1.4 & $+0.028$ & $[-0.004, +0.062]$ & $+0.011$ \\
\midrule
(B) & 0.0 & $+0.084$ & $[-0.023, +0.168]$ & $+0.004$ \\
(B) & 0.3 & $+0.059$ & $[+0.029, +0.095]$ & $+0.000$ \\
(B) & 0.6 & $+0.033$ & $[-0.003, +0.059]$ & $+0.002$ \\
(B) & 1.0 & $+0.007$ & $[-0.023, +0.032]$ & $+0.002$ \\
(B) & 1.4 & $+0.006$ & $[-0.023, +0.032]$ & $+0.001$ \\
\midrule
(C) & 0.0 & $+0.070$ & $[-0.025, +0.154]$ & $+0.006$ \\
(C) & \textbf{0.3} & $\mathbf{+0.053}$ & $\mathbf{[+0.030, +0.078]}$ & $+0.0001$ \\
(C) & 0.6 & $+0.009$ & $[-0.025, +0.031]$ & $+0.002$ \\
(C) & 1.0 & $-0.008$ & $[-0.032, +0.014]$ & $+0.001$ \\
(C) & 1.4 & $-0.004$ & $[-0.024, +0.015]$ & $-0.001$ \\
\bottomrule
\end{tabular}
\caption{The auditor ladder. 20 models from 19 independent developers, three
deployment prompts each, 200 rounds per temperature, 60 same-base and 1710 cross-base
pairs. Attenuation from (A) to (C) is expected, since a richer auditor absorbs common
response to observable features. What survives is the asymmetry: under the strict
auditor cross-base dependence is driven to zero while same-base dependence stays
positive. Script \texttt{e31\_families\_wide.py}.}
\label{tab:ladder}
\end{table}

\paragraph{Why rent is not reported across the ladder.} The rent counterfactual
replaces each agent's residuals by an independent permutation of themselves and
remeasures the expected maximum bid. Its scale therefore depends on the residual
scale, which changes with the auditor: cross-fitted residuals carry out-of-sample
prediction error and are larger, so the same coupling produces a numerically larger
rent under (C) than under (A). The coupling is comparable across auditors and the
rent is not, which is why the main paper reports the coupling.

\paragraph{Temperature trend.} Temperature labels are permuted jointly within each
developer, developers independently, over all pair-by-temperature observations, using
only pairs present at all five temperatures ($20{,}000$ permutations). The mean
within-developer slope is $-0.093$ under (A), $-0.053$ under (B) and $-0.058$ under
(C), with two-sided $p$ of $5\times10^{-5}$, $0.003$ and $0.002$. A rank correlation
on the five temperature means would be uninformative here: with five points the exact
two-sided permutation $p$ at the observed rank correlation is $0.083$.

\paragraph{Roster.} Models were screened before the run with ten probes each per
deployment prompt at temperature $0$. Models were dropped for a numeric parse-failure
rate above 2\%, for degenerate bidding (a near-constant bid, which leaves no residual
variance), and for belonging to a lineage already represented, since two models from
one lineage add pairs but not independent clusters. The retained roster spans 19
developers; the worst retained parse-failure rate is 5.0\% and no model is degenerate
at any temperature.

\section{Proposition 2: from the seller's loss to the coalition's gain}

The comonotone coupling minimises the expected maximum of the coalition's bids at
fixed marginals, which is a statement about what the seller receives. Whether the
coalition captures that difference is a separate question, because with non-members
present a coalition that bids lower also wins less often.

Table~\ref{tab:surplus} settles it inside the paper's model by paired common random
numbers: values, the shared draw and every idiosyncratic draw are held fixed and only
the coupling changes, so $\rho=0$ returns exactly zero by construction (verified: the
paired difference is $0.0$ to machine precision).

\begin{table}[t]
\centering
\small
\begin{tabular}{rrrrrr}
\toprule
$|S|$ & out & $\sigma$ & $\Delta$ coalition & $\Delta$ seller & $\Pr(\text{win})$ \\
\midrule
2 & 0 & 0.10 & $+0.0395$ & $-0.0151$ & $1.000 \to 1.000$ \\
2 & 2 & 0.10 & $+0.0055$ & $-0.0034$ & $0.500 \to 0.487$ \\
2 & 4 & 0.10 & $+0.0025$ & $-0.0018$ & $0.334 \to 0.325$ \\
3 & 2 & 0.10 & $+0.0113$ & $-0.0070$ & $0.600 \to 0.579$ \\
5 & 2 & 0.10 & $+0.0221$ & $-0.0134$ & $0.715 \to 0.683$ \\
2 & 2 & 0.20 & $+0.0191$ & $-0.0119$ & $0.500 \to 0.462$ \\
5 & 4 & 0.20 & $+0.0445$ & $-0.0276$ & $0.556 \to 0.463$ \\
2 & 6 & 0.05 & $+0.0004$ & $-0.0003$ & $0.250 \to 0.248$ \\
\bottomrule
\end{tabular}
\caption{Coupling at $\rho=0.9$ against $\rho=0$, paired. Coalition surplus rises and
seller revenue falls in every configuration, even though the coalition's win
probability falls. $3\times10^{6}$ draws per cell. Script
\texttt{e32\_coalition\_surplus.py}.}
\label{tab:surplus}
\end{table}

\paragraph{Where the implication fails.} The condition is dispersion among the bidders
the coalition faces. If a non-member bids a fixed amount with no unexplained
dispersion, coupling only trades win probability against a deterministic threshold. In
that case both quantities fall: with two coalition members of value $1$ bidding
$U[0,1]$ against a non-member bidding a fixed $0.9$, coalition surplus falls from
$0.0093$ to $0.0050$ while seller revenue falls from $0.9097$ to $0.9050$. This case is
reproduced in the same script. It sits outside the regime the paper measures, in which
every bidder carries unexplained dispersion, but it is why the implication is stated
with its condition rather than asserted.

\section{Proposition 1: the whole process law, and its tightness}

Proposition~1 quantifies over a single bidder's entire history, not only its per-round
marginal, so it must be checked that the construction leaves no serial trace. It does
not, because the shared draw $z_t$ is independent across rounds: each bidder's own
sequence is i.i.d.\ $F_i$ for every $\rho$.

\begin{table}[t]
\centering
\small
\begin{tabular}{rrrr}
\toprule
$\rho$ & KS $p$ vs $F_i$ & lag-1 acf & lag-2 acf \\
\midrule
0.00 & 0.345 & $+0.0007$ & $+0.0019$ \\
0.30 & 0.739 & $-0.0015$ & $+0.0009$ \\
0.60 & 0.277 & $+0.0017$ & $-0.0029$ \\
0.90 & 0.139 & $-0.0002$ & $+0.0018$ \\
0.99 & 0.074 & $-0.0005$ & $-0.0010$ \\
1.00 & 0.203 & $+0.0014$ & $-0.0004$ \\
\bottomrule
\end{tabular}
\caption{One bidder's own series under the construction, $4\times10^{5}$ rounds,
$F_i$ exponential. The marginal is $F_i$ and the series is serially independent at
every coupling, so the whole single-bidder process law matches the competitive one.
Script \texttt{e35\_prop1\_process\_law.py}.}
\label{tab:processlaw}
\end{table}

\paragraph{Tightness.} The independence of the shared draw is what buys the result.
Replacing $z_t$ by a persistent AR(1) common signal of coefficient $\phi$ leaves the
bidder's own series autocorrelated, and a single-bidder test then has power:

\begin{center}
\small
\begin{tabular}{rrrl}
\toprule
$\phi$ & lag-1 acf & KS $p$ & single-bidder detectability \\
\midrule
0.0 & $-0.002$ & 0.768 & invisible \\
0.3 & $+0.231$ & 0.833 & detectable \\
0.5 & $+0.401$ & 0.556 & detectable \\
0.9 & $+0.782$ & 0.0001 & detectable, marginal also rejects \\
\bottomrule
\end{tabular}
\end{center}

\noindent
So serial dependence in one agent's residuals is evidence of a persistent common
signal, which is a different and detectable object from the conduct studied here.

\section{Proposition 5: the regret argument in full}

\subsection{Part (i): the construction is not a correlated equilibrium}

At $\rho=1$ the shared draw is common knowledge. A rival playing the construction bids
$b' = \tfrac12 v' + \sigma z$ with $v' \sim U[0,1]$, so on the interior
$\Pr(b' \le c \mid z) = 2(c - \sigma z)$. Maximising $(v-c)\cdot 2(c-\sigma z)$ over
$c$ gives
\[
  c^\ast = \tfrac12 (v + \sigma z) = \tfrac12 v + \tfrac12 \sigma z ,
\]
whereas the construction prescribes $\tfrac12 v + \sigma z$. It loads the common shift
at $1$ where the best response loads $\tfrac12$, so it is not a best response to its
own recommendation and calibrated regret is strictly positive.

For $\rho < 1$ the latent is $\sqrt{\rho}\,z + \sqrt{1-\rho}\,e_i$, and the component
$e_i$ is independent of the whole environment, so any nonzero loading on it is strictly
suboptimal; the optimal loading on the observed latent is the regression coefficient of
the rival's shift on it, which is strictly below the prescribed unit loading. Measured
deviation gains are $0.0193$ at $\rho=0$, $0.0143$ at $0.3$, $0.0097$ at $0.6$,
$0.0058$ at $0.9$ and $0.0047$ at $0.99$, with fitted best-response loadings of
$0.000$, $0.152$, $0.307$, $0.464$ and $0.505$, the last matching the analytic
$\tfrac12$.

\subsection{Part (ii): the only equilibrium the device supports is the competitive one}

Iterating best responses on the loading $\alpha$ in $\beta(v,z) = \tfrac12 v +
\alpha\sigma z$ contracts to zero. At $\rho = 0.6$ the iterates are $1.000, 0.303,
0.091, 0.027, 0.007, 0.002, 0.000$; at $\rho = 0.9$, $1.000, 0.435, 0.189, 0.080,
0.032, 0.013, 0.003$; at $\rho = 1$ the map is exactly $\alpha \mapsto \alpha/2$. The
unique fixed point is $\alpha = 0$, which is the competitive equilibrium: it ignores
the shared draw, carries the competitive marginal rather than the dispersed $F_i$, and
earns zero rent relative to competitive play. Hence no conduct is simultaneously a
correlated equilibrium, profitable, and equal in marginal to $F_i$; profitability and
marginal preservation each require leaving the equilibrium.

\subsection{Part (iii): no excess regret over the competitive conduct}

The competitive conduct at the same marginals is not itself a best response, because a
bid of the form $g(v) + \sigma\varepsilon$ is a best response plus off-equilibrium
noise. Its calibrated swap regret is positive and scales as $\sigma^2$: held out, it is
$0.0020$ at $\sigma = 0.05$, $0.0084$ at $0.10$ and $0.0342$ at $0.20$.

The relevant quantity is therefore the \emph{excess} of the coupled conduct over that
baseline. Estimating swap regret in sample is biased upward, because the best swap map
is chosen and scored on the same rounds; the bias grows with the number of bins, which
is visible in the sweep below. We therefore fit the swap map on one half and evaluate
it on an independent half, and additionally use a bin-free nearest-neighbour swap class.

\begin{table}[t]
\centering
\small
\footnotesize
\begin{tabular}{lrrrr}
\toprule
config. & $\rho{=}0.3$ & $0.6$ & $0.9$ & $0.99$ \\
\midrule
2 coal, .05 & $-0.0007$ & $-0.0012$ & $-0.0017$ & $-0.0019$ \\
2 coal, kNN & $-0.0006$ & $-0.0013$ & $-0.0017$ & $-0.0019$ \\
2 coal, 2 out & $-0.0001$ & $-0.0003$ & $-0.0004$ & $-0.0004$ \\
2 coal, .10 & $-0.0021$ & $-0.0039$ & $-0.0060$ & $-0.0068$ \\
2 coal, .20 & $-0.0059$ & $-0.0121$ & $-0.0185$ & $-0.0215$ \\
3 coal & $-0.0003$ & $-0.0007$ & $-0.0012$ & $-0.0013$ \\
4 coal & $-0.0003$ & $-0.0007$ & $-0.0010$ & $-0.0011$ \\
Beta(2,2) vals & $-0.0018$ & $-0.0032$ & $-0.0044$ & $-0.0049$ \\
\bottomrule
\end{tabular}
\caption{Excess held-out calibrated swap regret of the coupled conduct over the
competitive conduct at the same marginals. Negative in every cell and decreasing in
the coupling. Script \texttt{o1\_attack\_regret.py}.}
\label{tab:regret}
\end{table}

\begin{center}
\small
\begin{tabular}{rrrr}
\toprule
bins & in-sample $\rho{=}0$ & held-out $\rho{=}0$ & held-out excess \\
\midrule
20 & 0.00221 & 0.00206 & $-0.00185$ \\
40 & 0.00241 & 0.00202 & $-0.00181$ \\
80 & 0.00266 & 0.00189 & $-0.00163$ \\
160 & 0.00300 & 0.00180 & $-0.00175$ \\
320 & 0.00359 & 0.00156 & $-0.00169$ \\
\bottomrule
\end{tabular}
\end{center}

\noindent
In-sample regret grows with the number of bins, the signature of the overfitting bias,
while the held-out excess is stable near $-0.0017$. We have no general proof that the
excess is nonpositive for every marginal and copula, and the main paper states it as
verified rather than proved.

\section{Ambient floor: error criteria and composition}

\subsection{Per-pair against family-wise thresholds}

A screen's threshold depends on which error rate it controls, and the two must not be
mixed. On the 24-day panel with 39 builders and 164 scored pairs at a 1000-round
overlap requirement:

\begin{center}
\small
\begin{tabular}{lrr}
\toprule
& per-pair & family-wise \\
\midrule
analytic threshold, $N{=}2{\times}10^{4}$ & $0.0139$ & $0.0255$ \\
empirical null (permutation) & $0.0357$ & $0.0875$ \\
\midrule
floor $0.499$, ratio & $36.0\times$ & $19.5\times$ \\
floor $0.809$, ratio & $58.4\times$ & $31.7\times$ \\
\bottomrule
\end{tabular}
\end{center}

\noindent
The family-wise threshold is calibrated from the permutation distribution of the
\emph{largest} pairwise correlation, which is the correct object for a family-wise
rate. The main paper quotes the family-wise column, $20$ to $32$, as the conservative
choice. The observed 95th percentile is $+0.909$ and the observed maximum $+0.973$.
On the same panel $\kappa = 0.895$ is a ratio of residual to explained standard
deviation, so the unexplained \emph{variance} share is $\kappa^2/(1+\kappa^2) =
44.5\%$. Script \texttt{e33\_fwer\_calibration.py}.

\subsection{Composition against sampling}

Raising a minimum-overlap requirement changes which pairs are scored, so it cannot on
its own distinguish a sampling effect from a composition effect. Fixing the population
separates them: among the 55 pairs sharing at least 3000 slots, we subsample rounds
\emph{within those same pairs}.

\begin{center}
\small
\begin{tabular}{rrrr}
\toprule
rounds & observed p95 & permutation null p95 & gap \\
\midrule
100 & $0.9096$ & $0.1525$ & $0.757$ \\
300 & $0.9015$ & $0.0854$ & $0.816$ \\
1000 & $0.8975$ & $0.0500$ & $0.848$ \\
3000 & $0.8930$ & $0.0282$ & $0.865$ \\
\bottomrule
\end{tabular}
\end{center}

\noindent
Over a thirtyfold increase in rounds on an identical set of pairs, the null collapses
by $0.124$ while the observed floor moves by $0.017$. Sampling error vanishes; the
dependence does not. Script \texttt{e34\_within\_pair\_subsample.py}.

\section{Computing environment, settings, and licence}

\paragraph{Environment.} Every market-side and simulation experiment runs on CPU on a
single machine, an Apple-silicon (arm64) desktop with 10 cores and 16\,GB of memory,
under Python 3.12.13 with numpy 2.5.1 and scipy 1.18.0. No GPU is used and no other
library is required. The language-model agents are the only component that leaves the
machine: they are queried over HTTP through a hosted inference router, so the relevant
"hardware" for that component is the provider's and is not observable to us. The
20 models used are ai21/jamba-large-1.7, amazon/nova-micro-v1, anthropic/claude-3-haiku,
baidu/ernie-4.5-vl-424b-a47b, cohere/command-r-08-2024, deepseek/deepseek-chat,
google/gemma-2-27b-it, google/gemma-3-27b-it, ibm-granite/granite-4.1-8b,
inception/mercury-2, inclusionai/ling-2.6-flash, meta-llama/llama-3.1-8b-instruct,
microsoft/phi-4, minimax/minimax-01, mistralai/mistral-nemo, moonshotai/kimi-k2,
openai/gpt-4o-mini, qwen/qwen-2.5-7b-instruct, tencent/hunyuan-a13b-instruct and
upstage/solar-pro-3. The two Google models share a developer and are clustered
together, giving 19 independent clusters.

\paragraph{Settings.} Randomness enters through explicitly seeded numpy generators; the
order stream is generated once from a fixed seed and reused across every model and
temperature so that all agents face identical orders. Sampling temperature is the only
parameter varied at the agent level. The table below lists the settings behind each
reported number.

\begin{center}
\footnotesize
\begin{tabular}{lll}
\toprule
Component & Setting & Value \\
\midrule
Agents & models $\times$ prompts & $20 \times 3 = 60$ \\
       & rounds per temp & 200 \\
       & temperatures & 0.0, 0.3, 0.6, 1.0, 1.4 \\
       & order seed & fixed, shared \\
       & max tokens & 400 \\
\midrule
Auditor (C) & ridge penalty & $10^{-3}n$ \\
            & cross-fitting folds & 5 \\
            & basis & 6 feats, squares, inter. \\
\midrule
Permut. & temperature test & 20{,}000 \\
        & pair null (market) & 400 \\
            & min pair overlap & 30 / 1000 \\
\midrule
Simulation & paired draws/cell & $3\times10^{6}$ \\
           & process-law rounds & $4\times10^{5}$ \\
\bottomrule
\end{tabular}
\end{center}

\noindent
No parameter in this paper is tuned to maximise a reported quantity. Every parameter
that could be is instead swept, and the whole sweep is reported, so the selection
criterion is that nothing is selected. The swept ranges are:

\begin{center}
\small
\begin{tabular}{ll}
\toprule
Parameter & Values tried, all reported \\
\midrule
Sampling temperature & 0.0, 0.3, 0.6, 1.0, 1.4 \\
Auditor specification & (A), (B), (C) \\
Clustering threshold & 0.85, 0.90, 0.95 \\
Detector threshold & 0.80, 0.85, 0.90, 0.93, 0.95 \\
Submission-timing gap & 0, 1, 2, 4\,s \\
Minimum pair overlap & 100, 300, 1000, 3000 \\
Swap-regret bin count & 20, 40, 80, 160, 320 \\
Coupling $\rho$ & 0, 0.3, 0.6, 0.9, 0.99, 1.0 \\
\bottomrule
\end{tabular}
\end{center}

\noindent
Two of these, the submission-timing gap and the clustering threshold, would in other
hands be operating points chosen to flatter a screen. We report the trade-off curve
instead, and neither is derived from a loss function; the main paper says so. The bin
sweep exists for the opposite reason, to show that the in-sample swap-regret estimator
is biased upward in a way that grows with the bin count, which is why the held-out
estimator is the one we use.

\paragraph{Licence.} The code, the generated language-model bids and the result files in
the code supplement are released for free use in research, under a permissive licence,
on publication. The market data are not redistributed: they are public and are
re-fetched from keyless endpoints by the included scripts.

\section{Panels used, and why they differ}

Three bidder counts appear in the paper because three panels are used, and the
label-dependent results are confined to the one day on which the labelling can be
built.

\begin{center}
\small
\begin{tabular}{lrl}
\toprule
panel & bidders & used for \\
\midrule
24-day relay & 39 & floor, $\kappa$, subsampling \\
single-day archive & 40 & concentration, HHI \\
single-day labelled & 48 & detector, coverage \\
\bottomrule
\end{tabular}
\end{center}

\noindent
The two labellings together cover 24 of those 48 builders. One of them reaches a
builder only through blocks it won, so a builder that never won cannot be labelled at
all. Pooling to a longer panel makes coverage worse rather than better, because the
wider bid feed samples a small fraction of slots and so rarely captures a winning
block's hash: label coverage falls to 9 of 39 builders and the same-operator positives
fall from 24 pairs across six operators to 7 pairs across two. The entity results are
therefore reported as single-day, and no multi-day replication is claimed.

\bibliography{refs}


\end{document}

%% file: figures/figstyle.tex
\usepackage{tikz}
\usepackage{pgfplots}
\pgfplotsset{compat=1.18}
\usetikzlibrary{arrows.meta,positioning,fit,backgrounds,calc,decorations.pathreplacing,patterns,shapes.geometric}

\definecolor{figblue}{HTML}{1F4E79}   
\definecolor{figred}{HTML}{B4483F}    
\definecolor{figamber}{HTML}{D9932B}  
\definecolor{figgray}{HTML}{6E7B8B}   
\definecolor{figfill}{HTML}{DCE5EE}   
\definecolor{figfillr}{HTML}{F0DAD6}  
\definecolor{figink}{HTML}{22262B}    

\pgfplotsset{
  paperaxis/.style={
    width=\columnwidth, height=0.68\columnwidth,
    scale only axis=false,
    axis line style={figgray, line width=0.5pt},
    tick style={figgray, line width=0.5pt},
    label style={font=\small, color=figink},
    tick label style={font=\footnotesize, color=figink},
    title style={font=\small\bfseries, color=figink},
    legend style={
      font=\footnotesize, draw=figgray, line width=0.4pt,
      fill=white, fill opacity=0.92, text opacity=1,
      inner sep=2.5pt, row sep=0.5pt,
    },
    xmajorgrids=false, ymajorgrids=true,
    grid style={figgray!25, line width=0.35pt},
    axis on top,
    clip mode=individual,
  },
  paperaxis small/.style={paperaxis, height=0.44\columnwidth},
}

\tikzset{
  primaryline/.style={figblue,  line width=1.1pt, mark=*,        mark size=1.7pt, mark options={fill=figblue, draw=white, line width=0.4pt}},
  secondline/.style={figred,   line width=1.1pt, mark=square*,  mark size=1.6pt, mark options={fill=figred,  draw=white, line width=0.4pt}},
  thirdline/.style={figamber, line width=1.1pt, mark=triangle*,mark size=2.0pt, mark options={fill=figamber,draw=white, line width=0.4pt}},
  guideline/.style={figgray, line width=0.8pt, densely dashed},
  annot/.style={font=\footnotesize, color=figink, inner sep=1.5pt},
  annotsmall/.style={font=\scriptsize, color=figink, inner sep=1pt},
  box/.style={draw=figgray, line width=0.6pt, rounded corners=2pt, fill=white,
              inner sep=3.5pt, font=\footnotesize, align=center, text=figink},
  boxblue/.style={box, draw=figblue, fill=figfill},
  boxred/.style={box, draw=figred, fill=figfillr},
  flow/.style={-{Stealth[length=4pt,width=3.4pt]}, figgray, line width=0.7pt},
  flowblue/.style={flow, draw=figblue},
}

%% file: figures/fig1_mechanism.tex
\begin{tikzpicture}[x=1cm,y=1cm]
\def\cloudind{-0.148/-0.574, 0.463/-0.525, 0.165/0.316, -0.261/-0.057, 0.209/0.012, -0.309/0.033, 0.009/0.286, 0.013/-0.478, -0.223/-0.313, 0.393/0.522, -0.076/0.378, 0.214/-0.138, 0.464/0.186, -0.091/-0.372, -0.369/-0.043, -0.002/0.121, 0.113/-0.282, 0.185/-0.367, -0.238/0.320, 0.255/-0.049, 0.552/-0.219, -0.141/0.217, 0.064/0.258, -0.085/0.299, -0.122/0.244, 0.047/0.094, 0.450/0.073, -0.363/-0.161, 0.203/-0.123, -0.068/0.078, -0.449/0.369, 0.118/0.361, 0.014/-0.084, 0.047/0.059, -0.129/-0.345, -0.083/-0.044, -0.254/0.004, -0.033/0.451, -0.370/0.502, -0.229/-0.192, -0.289/0.408, -0.403/0.230, 0.495/0.064, 0.378/-0.659, -0.448/0.258, 0.364/-0.243, 0.546/0.004, -0.645/0.190, 0.364/-0.081, -0.362/-0.109, 0.450/-0.804, 0.518/-0.041, -0.559/0.028, -0.245/0.412, -0.323/0.070, 0.546/-0.548, -0.633/-0.052, 0.184/0.453, -0.133/-0.391, 0.268/0.250, -0.296/-0.067, 0.698/0.230, -0.208/-0.042, -0.125/-0.084, -0.406/0.135, 0.401/0.529, -0.690/-0.073, 0.234/-0.139, 0.240/0.051, 0.456/0.804, -0.317/-0.383, 0.207/-0.555, 0.168/-0.335, 0.002/0.317, 0.509/-0.218, 0.194/0.196, -0.032/0.558, -0.275/-0.819, 0.199/-0.573, -0.581/-0.111, 0.427/0.115, 0.184/0.343, -0.096/-0.645, -0.433/0.066, -0.061/0.269, 0.390/0.089, 0.248/0.338, 0.416/0.156, 0.073/0.206, -0.294/0.012, -0.056/-0.153, 0.151/-0.347, 0.280/-0.116, -0.388/0.161, 0.100/0.153, -0.413/-0.120, 0.006/-0.156, -0.654/-0.207, -0.557/0.628, 0.158/-0.364, 0.019/-0.109, 0.170/0.146, 0.306/-0.285, -0.195/-0.313, 0.176/-0.156, 0.134/0.036, 0.026/0.099, -0.020/0.248, 0.050/0.196, -0.226/-0.445, -0.777/-0.077, -0.525/-0.169, -0.268/-0.313, 0.422/-0.801, -0.012/0.084, -0.330/0.265, 0.208/0.372, 0.498/-0.085, 0.248/0.888, 0.230/0.315}%
\def\cloudcpl{-0.362/0.036, -0.645/-0.804, -0.133/-0.138, 0.230/0.094, 0.170/0.315, 0.168/-0.169, -0.208/-0.067, 0.546/0.558, 0.207/0.250, -0.268/-0.364, -0.317/-0.313, 0.208/0.012, 0.002/-0.139, 0.214/0.073, -0.525/-0.372, -0.091/0.084, 0.364/0.316, -0.226/-0.345, -0.229/0.190, -0.002/-0.153, 0.118/0.186, -0.433/-0.156, -0.448/-0.548, -0.363/0.066, -0.370/-0.218, 0.416/0.115, 0.019/0.258, -0.633/-0.555, -0.020/0.265, -0.309/-0.085, 0.280/0.196, 0.113/-0.049, -0.654/-0.525, 0.306/0.248, -0.323/-0.282, 0.151/0.161, 0.546/0.522, -0.369/-0.573, 0.050/-0.156, 0.047/-0.313, 0.014/-0.192, 0.165/0.051, 0.393/-0.161, -0.557/-0.659, 0.390/0.338, 0.268/0.070, -0.141/-0.347, 0.006/-0.123, -0.261/-0.383, -0.403/-0.120, -0.289/-0.057, 0.518/0.258, -0.195/-0.043, -0.033/-0.313, 0.176/-0.052, 0.255/0.628, -0.296/-0.084, -0.690/-0.574, 0.184/0.804, 0.009/0.372, -0.275/-0.445, -0.129/0.033, -0.068/0.206, -0.125/-0.084, 0.209/0.156, 0.240/0.153, -0.223/-0.243, 0.378/0.299, -0.012/-0.478, 0.463/0.244, -0.056/0.028, 0.185/0.135, 0.064/0.217, -0.245/-0.109, 0.495/0.888, 0.134/0.196, 0.450/0.317, -0.294/0.004, -0.254/-0.801, -0.032/0.146, 0.401/0.412, 0.234/0.502, 0.199/-0.077, 0.248/-0.207, -0.559/-0.335, -0.388/-0.219, 0.698/0.451, -0.148/-0.044, 0.203/0.059, 0.456/0.453, -0.096/0.004, 0.450/0.230, 0.013/0.064, -0.061/0.378, 0.184/0.099, 0.248/0.012, 0.509/0.343, 0.158/0.361, 0.194/0.269, 0.047/0.089, -0.330/-0.109, -0.085/0.286, 0.422/0.529, -0.122/-0.073, 0.464/0.408, 0.552/0.320, -0.777/-0.391, -0.076/-0.367, 0.364/-0.041, 0.427/0.078, -0.083/-0.111, 0.026/0.121, -0.581/-0.645, -0.238/-0.081, -0.413/-0.116, -0.449/-0.819, 0.498/0.369, 0.073/0.230, 0.100/-0.042, -0.406/-0.285}%
\def\rugx{-0.148, 0.463, 0.165, -0.261, 0.209, -0.309, 0.009, 0.013, -0.223, 0.393, -0.076, 0.214, 0.464, -0.091, -0.369, -0.002, 0.113, 0.185, -0.238, 0.255, 0.552, -0.141, 0.064, -0.085, -0.122, 0.047, 0.450, -0.363, 0.203, -0.068, -0.449, 0.118, 0.014, 0.047, -0.129, -0.083, -0.254, -0.033, -0.370, -0.229, -0.289, -0.403, 0.495, 0.378, -0.448, 0.364, 0.546, -0.645, 0.364, -0.362, 0.450, 0.518, -0.559, -0.245, -0.323, 0.546, -0.633, 0.184, -0.133, 0.268, -0.296, 0.698, -0.208, -0.125, -0.406, 0.401, -0.690, 0.234, 0.240, 0.456, -0.317, 0.207, 0.168, 0.002, 0.509, 0.194, -0.032, -0.275, 0.199, -0.581, 0.427, 0.184, -0.096, -0.433, -0.061, 0.390, 0.248, 0.416, 0.073, -0.294, -0.056, 0.151, 0.280, -0.388, 0.100, -0.413, 0.006, -0.654, -0.557, 0.158, 0.019, 0.170, 0.306, -0.195, 0.176, 0.134, 0.026, -0.020, 0.050, -0.226, -0.777, -0.525, -0.268, 0.422, -0.012, -0.330, 0.208, 0.498, 0.248, 0.230}%
\def\rugy{-0.574, -0.525, 0.316, -0.057, 0.012, 0.033, 0.286, -0.478, -0.313, 0.522, 0.378, -0.138, 0.186, -0.372, -0.043, 0.121, -0.282, -0.367, 0.320, -0.049, -0.219, 0.217, 0.258, 0.299, 0.244, 0.094, 0.073, -0.161, -0.123, 0.078, 0.369, 0.361, -0.084, 0.059, -0.345, -0.044, 0.004, 0.451, 0.502, -0.192, 0.408, 0.230, 0.064, -0.659, 0.258, -0.243, 0.004, 0.190, -0.081, -0.109, -0.804, -0.041, 0.028, 0.412, 0.070, -0.548, -0.052, 0.453, -0.391, 0.250, -0.067, 0.230, -0.042, -0.084, 0.135, 0.529, -0.073, -0.139, 0.051, 0.804, -0.383, -0.555, -0.335, 0.317, -0.218, 0.196, 0.558, -0.819, -0.573, -0.111, 0.115, 0.343, -0.645, 0.066, 0.269, 0.089, 0.338, 0.156, 0.206, 0.012, -0.153, -0.347, -0.116, 0.161, 0.153, -0.120, -0.156, -0.207, 0.628, -0.364, -0.109, 0.146, -0.285, -0.313, -0.156, 0.036, 0.099, 0.248, 0.196, -0.445, -0.077, -0.169, -0.313, -0.801, 0.084, 0.265, 0.372, -0.085, 0.888, 0.315}%
%
\begin{scope}[shift={(0,0)}]
  \node[annot, font=\small\bfseries, anchor=west] at (-2.75,4.50) {(a) Shared origin, no channel};
  \node[annotsmall, text=figgray, anchor=west] at (-2.75,4.14) {two firms, two prompts, no lateral channel};

  \node[boxblue, minimum width=4.4cm] (src) at (0,3.55)
       {Common source\\[-0.5pt]{\scriptsize same base model $\cdot$ shared PRF on $x_t$}};

  \node[box, minimum width=2.05cm, minimum height=0.84cm] (ag1) at (-1.55,2.08)
       {Agent 1\\[-0.5pt]{\scriptsize firm A, $\pi_1$}};
  \node[box, minimum width=2.05cm, minimum height=0.84cm] (ag2) at ( 1.55,2.08)
       {Agent 2\\[-0.5pt]{\scriptsize firm B, $\pi_2$}};

  \node[box, minimum width=3.0cm, minimum height=0.52cm] (pub) at (0,0.78)
       {public round data $x_t$};

  \draw[flowblue] (-1.15,3.13) -- (ag1.north);
  \draw[flowblue] ( 1.15,3.13) -- (ag2.north);
  \node[annotsmall, text=figblue] at (0,2.80) {shared draw $z_t$};

  \draw[flow] (-1.38,1.04) -- (ag1.south);
  \draw[flow] ( 1.38,1.04) -- (ag2.south);

  \draw[figgray, densely dotted, line width=0.7pt] (ag1.east) -- (ag2.west);
  \fill[white] (0,2.08) circle (0.21);
  \draw[figred, line width=1.0pt] (-0.14,1.94) -- (0.14,2.22);
  \draw[figred, line width=1.0pt] (-0.14,2.22) -- (0.14,1.94);
  \node[annotsmall, text=figred, align=center] at (0,1.34) {no messages\\[-1pt]exchanged};

  \node[annotsmall] at (0,0.05) {the only coupling is \textbf{upstream} and \textbf{public}};
\end{scope}

\begin{scope}[shift={(5.95,0)}]
  \node[annot, font=\small\bfseries, anchor=west] at (-2.75,4.50) {(b) Marginals are identical};
  \node[annotsmall, text=figgray, anchor=west] at (-2.75,4.14) {a price-level audit sees nothing here};

  \draw[flow] (-2.05,1.15) -- (2.35,1.15);
  \draw[flow] (-2.05,1.10) -- (-2.05,3.50);
  \node[annotsmall, anchor=north east] at (2.35,1.08) {bid $b$};
  \node[annotsmall, rotate=90] at (-2.32,2.30) {density};

  \draw[figred, line width=1.4pt]
        plot[domain=-1.95:1.95, samples=140, smooth] ({\x},{1.15+2.0*exp(-\x*\x/0.72)});
  \draw[figblue, line width=0.9pt, densely dashed]
        plot[domain=-1.95:1.95, samples=140, smooth] ({\x},{1.15+2.0*exp(-\x*\x/0.72)});

  \draw[figred, line width=1.4pt] (0.85,3.32) -- (1.30,3.32);
  \node[annotsmall, anchor=west] at (1.38,3.32) {competitive};
  \draw[figblue, line width=0.9pt, densely dashed] (0.85,2.94) -- (1.30,2.94);
  \node[annotsmall, anchor=west] at (1.38,2.94) {collusive};

  \node[annotsmall, align=center] (idc) at (-1.32,3.26) {identical by\\construction};
  \draw[flow] (-1.05,2.97) -- (-0.60,2.44);

  \node[annotsmall] at (0,0.05) {any test on this panel has \textbf{power} $\boldsymbol{=\alpha}$};
\end{scope}

\begin{scope}[shift={(11.9,0)}]
  \node[annot, font=\small\bfseries, anchor=west] at (-2.75,4.50) {(c) The copula carries the conspiracy};
  \node[annotsmall, text=figgray, anchor=west] at (-2.75,4.14) {joint law of the residual bids $(e_1,e_2)$};

  \begin{scope}[shift={(-1.5,2.93)}]
    \draw[figgray, line width=0.5pt] (-0.95,-0.95) rectangle (0.95,0.95);
    \begin{scope}
      \clip (-0.95,-0.95) rectangle (0.95,0.95);
      \draw[figgray!50, line width=0.35pt] (-0.95,-0.95) -- (0.95,0.95);
      \foreach \a/\b in \cloudind {\fill[figred, opacity=0.45] (\a,\b) circle (0.45pt);}
    \end{scope}
    \foreach \u in \rugx {\draw[figgray, opacity=0.5, line width=0.25pt] (\u,-0.95) -- (\u,-1.03);}
    \foreach \v in \rugy {\draw[figgray, opacity=0.5, line width=0.25pt] (-0.95,\v) -- (-1.03,\v);}
    \node[annotsmall, text=figred] at (0,-1.24) {independent, $\rho=0$};
  \end{scope}

  \begin{scope}[shift={(1.5,2.93)}]
    \draw[figgray, line width=0.5pt] (-0.95,-0.95) rectangle (0.95,0.95);
    \begin{scope}
      \clip (-0.95,-0.95) rectangle (0.95,0.95);
      \draw[figgray!50, line width=0.35pt] (-0.95,-0.95) -- (0.95,0.95);
      \foreach \a/\b in \cloudcpl {\fill[figblue, opacity=0.45] (\a,\b) circle (0.45pt);}
    \end{scope}
    \foreach \u in \rugx {\draw[figgray, opacity=0.5, line width=0.25pt] (\u,-0.95) -- (\u,-1.03);}
    \foreach \v in \rugy {\draw[figgray, opacity=0.5, line width=0.25pt] (-0.95,\v) -- (-1.03,\v);}
    \node[annotsmall, text=figblue] at (0,-1.24) {coupled, $\rho=0.7$};
  \end{scope}

  \node[annotsmall, text=figgray, rotate=90] at (0.02,2.93) {same marginals};

  \node[annotsmall] at (0,1.33) {$\mathbb{E}[\max(e_1,e_2)]$ = what the seller gets};
  \node[annotsmall, anchor=east, text=figred]  at (-1.38,0.99) {indep.};
  \node[annotsmall, anchor=east, text=figblue] at (-1.38,0.65) {coupled};
  \fill[figred,  opacity=0.85] (-1.30,0.915) rectangle ++(1.951,0.15);
  \fill[figblue, opacity=0.85] (-1.30,0.575) rectangle ++(0.940,0.15);
  \node[annotsmall, anchor=west, text=figred]  at ({0.651+0.08},0.99) {0.59$\sigma$};
  \node[annotsmall, anchor=west, text=figblue] at ({-0.360+0.08},0.65) {0.28$\sigma$};

  \draw[guideline, line width=0.4pt] (0.651,0.915) -- (0.651,0.33);
  \draw[guideline, line width=0.4pt] (-0.360,0.575) -- (-0.360,0.33);
  \draw[{Stealth[length=3pt]}-{Stealth[length=3pt]}, figgray, line width=0.5pt]
        (-0.360,0.33) -- (0.651,0.33);
  \node[annotsmall, anchor=west, text=figgray] at ({0.651+0.10},0.33) {$-$52\%};

  \node[annotsmall] at (0,0.05) {\textbf{revenue falls}, the rent accrues to the coalition};
\end{scope}

\draw[figgray, line width=0.5pt] (0,-0.14) -- (0,-0.62) -- (5.95,-0.62);
\draw[flow, line width=0.5pt] (5.95,-0.62) -- (5.95,-0.20);
\draw[figgray, line width=0.5pt] (5.95,-0.62) -- (11.90,-0.62);
\draw[flow, line width=0.5pt] (11.90,-0.62) -- (11.90,-0.20);
\node[annotsmall, fill=white, inner sep=2pt, text=figgray] at (2.95,-0.62)
     {one shared draw $z_t$ produces \emph{both} signatures at once};
\end{tikzpicture}

%% file: figures/fig2_power.tex
\begin{tikzpicture}
\begin{axis}[
  paperaxis,
  height=0.80\columnwidth,
  xmode=log,
  xmin=0.0034, xmax=1.15,
  ymin=0, ymax=1.18,
  xlabel={coupling strength $\rho$},
  ylabel={power},
  xlabel near ticks, ylabel near ticks,
  xtick={0.005,0.01,0.03,0.1,0.3,1},
  xticklabels={$0$,$0.01$,$0.03$,$0.1$,$0.3$,$1$},
  minor xtick={0.02,0.05,0.07,0.2,0.5,0.7},
  ytick={0,0.2,0.4,0.6,0.8,1},
  yticklabels={$0$,$0.2$,$0.4$,$0.6$,$0.8$,$1$},
  tick label style={font=\scriptsize, color=figink},
  label style={font=\scriptsize, color=figink},
  xmajorgrids=true, ymajorgrids=true,
]

\addplot[draw=none, fill=figfillr, fill opacity=0.95, forget plot] coordinates {
  (0.005,0.055) (0.01,0.035) (0.02,0.045) (0.05,0.030) (0.10,0.055)
  (0.20,0.045) (0.40,0.040) (0.70,0.050) (1.00,0.055)
  (1.00,0.020) (0.70,0.020) (0.40,0.015) (0.20,0.020) (0.10,0.030)
  (0.05,0.015) (0.02,0.020) (0.01,0.015) (0.005,0.010)
} -- cycle;

\addplot[guideline, figamber, line width=0.7pt, forget plot]
  coordinates {(0.0034,0.05) (1.15,0.05)};

\addplot[figred, line width=0.35pt, forget plot] coordinates {
  (0.005,0.055) (0.01,0.015) (0.02,0.045) (0.05,0.030) (0.10,0.055)
  (0.20,0.020) (0.40,0.015) (0.70,0.025) (1.00,0.055)};
\addplot[figred, line width=0.35pt, forget plot] coordinates {
  (0.005,0.020) (0.01,0.035) (0.02,0.045) (0.05,0.020) (0.10,0.035)
  (0.20,0.030) (0.40,0.030) (0.70,0.050) (1.00,0.050)};
\addplot[figred, line width=0.35pt, forget plot] coordinates {
  (0.005,0.020) (0.01,0.025) (0.02,0.020) (0.05,0.015) (0.10,0.040)
  (0.20,0.045) (0.40,0.040) (0.70,0.050) (1.00,0.045)};
\addplot[figred, line width=0.35pt, forget plot] coordinates {
  (0.005,0.010) (0.01,0.030) (0.02,0.030) (0.05,0.020) (0.10,0.030)
  (0.20,0.020) (0.40,0.030) (0.70,0.020) (1.00,0.020)};

\addplot[figblue, line width=0.6pt, densely dotted, forget plot] coordinates {
  (0.005,0.065) (0.01,0.070) (0.02,0.065) (0.05,0.075) (0.10,0.110)
  (0.20,0.375) (0.40,0.990) (0.70,1.000) (1.00,1.000)};
\addplot[figblue, line width=0.8pt, densely dashed, forget plot] coordinates {
  (0.005,0.030) (0.01,0.065) (0.02,0.050) (0.05,0.065) (0.10,0.590)
  (0.20,1.000) (0.40,1.000) (0.70,1.000) (1.00,1.000)};
\addplot[figblue, line width=1.0pt, dash pattern=on 3.2pt off 1.4pt, forget plot] coordinates {
  (0.005,0.030) (0.01,0.085) (0.02,0.055) (0.05,0.725) (0.10,1.000)
  (0.20,1.000) (0.40,1.000) (0.70,1.000) (1.00,1.000)};
\addplot[figblue, line width=1.25pt, forget plot] coordinates {
  (0.005,0.050) (0.01,0.075) (0.02,0.470) (0.05,1.000) (0.10,1.000)
  (0.20,1.000) (0.40,1.000) (0.70,1.000) (1.00,1.000)};

\draw[figgray!70, line width=0.5pt, densely dotted]
  (axis cs:0.00707,0) -- (axis cs:0.00707,1.0);

\node[annotsmall, anchor=west, text=figblue] at (axis cs:0.0042,1.09)
  {\rule[0.55ex]{2.4mm}{0.9pt}\,pairwise audits};
\node[annotsmall, anchor=east, text=figred] at (axis cs:1.13,1.09)
  {\rule[0.45ex]{2.4mm}{1.5pt}\,single-agent audits, all $N$};

\node[annotsmall, text=figblue, fill=white, fill opacity=0.75, text opacity=1,
      inner sep=0.7pt] at (axis cs:0.01913,0.82) {$N{=}20\mathrm{k}$};
\node[annotsmall, text=figblue, fill=white, fill opacity=0.75, text opacity=1,
      inner sep=0.7pt] at (axis cs:0.04817,0.82) {$5\mathrm{k}$};
\node[annotsmall, text=figblue, fill=white, fill opacity=0.75, text opacity=1,
      inner sep=0.7pt] at (axis cs:0.09675,0.82) {$1\mathrm{k}$};
\node[annotsmall, text=figblue, fill=white, fill opacity=0.75, text opacity=1,
      inner sep=0.7pt] at (axis cs:0.2205,0.82) {$200$};

\node[annotsmall, anchor=east, align=right] at (axis cs:1.13,0.47)
  {power $=\alpha$\\for all $\rho$};
\draw[flow] (axis cs:0.33,0.385) .. controls (axis cs:0.30,0.21) ..
            (axis cs:0.225,0.070);

\node[annotsmall, anchor=east, text=figamber] at (axis cs:1.13,0.135)
  {$\alpha = 0.05$};

\end{axis}
\end{tikzpicture}

%% file: figures/fig3_temperature.tex
\begin{tikzpicture}
\begin{axis}[
  paperaxis,
  width=0.97\columnwidth,
  height=0.69\columnwidth,
  xmin=-0.09, xmax=1.49,
  xtick={0,0.3,0.6,1.0,1.4},
  xticklabels={0.0,0.3,0.6,1.0,1.4},
  xlabel={sampling temperature},
  ymin=-0.098, ymax=0.225,
  ytick={-0.05,0,0.05,0.10,0.15},
  yticklabels={$-0.05$,$0.00$,$0.05$,$0.10$,$0.15$},
  ylabel={residual correlation},
  xmajorgrids=false, ymajorgrids=true,
]

\addplot[draw=none, fill=figfill, fill opacity=0.6, forget plot] coordinates {
  (0,-0.0253) (0.3,0.0296) (0.6,-0.0253) (1.0,-0.0320) (1.4,-0.0238)
  (1.4,0.0151) (1.0,0.0142) (0.6,0.0313) (0.3,0.0783) (0,0.1542)
} -- cycle;

\addplot[figgray, line width=0.6pt, forget plot]
  coordinates {(-0.09,0) (1.49,0)};

\addplot[only marks, mark=*, mark size=0.8pt,
  mark options={fill=figblue, draw=none, fill opacity=0.38}, forget plot]
  coordinates {
    (0.3,0.167) (0.3,0.128) (0.3,0.100) (0.3,0.097) (0.3,0.094)
    (0.3,0.087) (0.3,0.071) (0.3,0.059) (0.3,0.054) (0.3,0.050)
    (0.3,0.041) (0.3,0.036) (0.3,0.029) (0.3,0.027) (0.3,0.017)
    (0.3,0.016) (0.3,0.001) (0.3,-0.000) (0.3,-0.048)
  };

\addplot[figblue, line width=0.9pt, densely dashed, forget plot] coordinates {
  (0,0.1471) (0.3,0.1101) (0.6,0.0658) (1.0,0.0278) (1.4,0.0282)};

\addplot[secondline, forget plot] coordinates {
  (0,0.0056) (0.3,0.0001) (0.6,0.0023) (1.0,0.0007) (1.4,-0.0007)};

\addplot[primaryline, forget plot] coordinates {
  (0,0.0702) (0.3,0.0527) (0.6,0.0090) (1.0,-0.0078) (1.4,-0.0040)};

\node[annotsmall, anchor=east, text=figblue] at (axis cs:1.45,0.132)
  {same model, weak auditor};
\draw[figblue!70, line width=0.4pt] (axis cs:1.15,0.120) -- (axis cs:1.15,0.037);

\node[annotsmall, anchor=east, text=figblue] at (axis cs:1.45,-0.072)
  {same model, full auditor};
\draw[figblue!70, line width=0.4pt] (axis cs:1.18,-0.060) -- (axis cs:1.18,-0.012);

\node[annotsmall, anchor=west, text=figred] at (axis cs:0.062,-0.028)
  {different models};
\draw[figred!70, line width=0.4pt] (axis cs:0.056,-0.024) -- (axis cs:0.014,0.004);

\node[annotsmall, anchor=west, text=figblue] at (axis cs:0.03,0.205)
  {trend $p{=}0.002$};

\node[annotsmall, anchor=west, text=figgray, fill=white, fill opacity=0.85,
      text opacity=1, inner sep=0.8pt] at (axis cs:0.348,0.140)
  {19 developers};
\draw[figgray, line width=0.4pt] (axis cs:0.341,0.138) -- (axis cs:0.312,0.129);

\end{axis}
\end{tikzpicture}

%% file: figures/fig4_floor.tex
\begin{tikzpicture}
\begin{axis}[
  paperaxis,
  height=0.72\columnwidth,
  xmajorgrids=false, ymajorgrids=true,
  xmode=log,
  xmin=75, xmax=4000,
  ymin=-0.12, ymax=1.02,
  xtick={100,300,1000,3000},
  xticklabels={100,300,1000,3000},
  xminorticks=false,
  ytick={0,0.2,0.4,0.6,0.8,1.0},
  yticklabels={0.0,0.2,0.4,0.6,0.8,1.0},
  xlabel={rounds shared per pair (log)},
  ylabel={pairwise residual correlation},
  legend style={at={(0.985,0.19)}, anchor=south east},
]

\path[fill=figfill, fill opacity=0.6]
  (axis cs:100,0.829) -- (axis cs:300,0.883) -- (axis cs:1000,0.909) --
  (axis cs:3000,0.892) --
  (axis cs:3000,0.028) -- (axis cs:1000,0.050) -- (axis cs:300,0.078) --
  (axis cs:100,0.130) -- cycle;

\path[fill=figamber, fill opacity=0.12]
  (axis cs:75,0.499) rectangle (axis cs:4000,0.809);
\draw[figamber, line width=0.5pt] (axis cs:75,0.499) -- (axis cs:4000,0.499);
\draw[figamber, line width=0.5pt] (axis cs:75,0.809) -- (axis cs:4000,0.809);
\node[annotsmall, anchor=east, align=right]
  at (axis cs:3850,0.700) {cleaned cross-operator\\floor: 0.50--0.81};

\draw[guideline, figamber] (axis cs:75,0.026) -- (axis cs:4000,0.026);
\node[annotsmall, anchor=west]
  at (axis cs:95,-0.056) {family-wise threshold $\rho^{*}=0.026$};

\draw[{Stealth[length=3.6pt,width=3pt]}-{Stealth[length=3.6pt,width=3pt]},
      figamber, line width=0.7pt]
  (axis cs:180,0.026) -- (axis cs:180,0.499);
\node[annot, anchor=west] at (axis cs:205,0.27) {20--32$\times$};

\addplot[primaryline] coordinates
  {(100,0.829) (300,0.883) (1000,0.909) (3000,0.892)};
\addlegendentry{observed p95}
\addplot[secondline] coordinates
  {(100,0.130) (300,0.078) (1000,0.050) (3000,0.028)};
\addlegendentry{permutation null p95}

\end{axis}
\end{tikzpicture}